\documentclass[a4paper,UKenglish,cleveref, autoref, thm-restate]{lipics-v2021}
\usepackage{float}
\usepackage{tikz} 
\usetikzlibrary{decorations.markings}
\usepackage[justification=centering]{caption}

\title{NP-Completeness Proofs of All or Nothing, Water Walk, and Remembered Length Using the T-Metacell Framework}
\titlerunning{All or Nothing, Water Walk, and Remembered Length are NP-Complete}

\author{Pakapim Eua-anant}{Department of Computer Engineering, Chulalongkorn University, Bangkok, Thailand}{pakpim.e@gmail.com}{}{}

\author{Papangkorn Apinyanon}{St. Stephen's Episcopal School, TX, USA}{apinyanonp@gmail.com}{}{}

\author{Thunyatorn Jirachaisri}{New Hampton School, NH, USA}{tjnjtjcj@gmail.com}{}{}

\author{Nantapong Ruangsuksriwong}{Department of Computer Engineering, Chulalongkorn University, Bangkok, Thailand}{teamnantapong@gmail.com}{}{}

\author{Suthee Ruangwises}{Department of Computer Engineering, Chulalongkorn University, Bangkok, Thailand}{suthee@cp.eng.chula.ac.th}{}{}

\authorrunning{P. Eua-anant, P. Apinyanon, T. Jirachaisri, N. Ruangsuksriwong, S. Ruangwises}

\Copyright{P. Eua-anant, P. Apinyanon, T. Jirachaisri, N. Ruangsuksriwong, S. Ruangwises} 



\category{} 

\EventEditors{}
\EventNoEds{0}
\EventLongTitle{}
\EventShortTitle{}
\EventAcronym{}
\EventYear{}
\EventDate{}
\EventLocation{}
\EventLogo{}
\SeriesVolume{}
\ArticleNo{}

\tikzset{
  pathline/.style={
    red,
    very thick,
    cap=round,
    postaction={
      decorate,
      decoration={
        markings,
        mark=between positions 0.05 and 0.95 step 18pt
        with {\arrow{Stealth[length=4pt,width=5pt]}}
      }
    }
  }
}

\makeatletter
\newcommand*{\rom}[1]{\expandafter\@slowromancap\romannumeral #1@}
\makeatother

\begin{document}
\maketitle

\begin{abstract}
  All or Nothing, Water Walk, and Remembered Length are pencil puzzles that involve constructing a continuous loop on a rectangular grid under specific constraints. In this paper, we analyze their computational complexity using the T-metacell framework developed by Tang~\cite{tang2022frameworklooppathpuzzle} and MIT Hardness Group~\cite{37puzzles}. We establish that the puzzles are NP-complete by providing reductions; the first two puzzles, from the problem of finding a Hamiltonian cycle in a maximum-degree-3 spanning subgraph of a rectangular grid graph, and the last, from the problem of finding a Hamiltonian cycle in a required-edge directed rectangular grid graph.
\end{abstract}

\section{Introduction}
Pencil puzzles such as Sudoku, Nonogram, and Numberlink are logic-based puzzles designed to be solved on paper through deductive reasoning. Beyond their recreational appeal, these puzzles hold theoretical significance in computational complexity, as they exemplify problems that are simple to state yet hard to solve. A large number of pencil puzzles have been proven NP-complete through various reduction frameworks, including Chained Block~\cite{chained}, Choco Banana~\cite{choco}, Five Cells~\cite{fivecells}, Hashiwokakero~\cite{bridges}, Kakuro~\cite{sudoku}, Numberlink~\cite{numberlink}, Roma~\cite{roma}, Slitherlink~\cite{sudoku}, Sudoku~\cite{sudoku}, Suguru~\cite{suguru}, Sumplete~\cite{sumplete}, Tatamibari~\cite{tatamibari}, Tilepaint~\cite{fivecells}, and Zeiger~\cite{zeiger}.




\subsection{T-Metacell}
Numerous pencil puzzles that involve constructing a single continuous loop have been shown to be computationally hard through reductions from various forms of Hamiltonicity problems. In 2022, Tang~\cite{tang2022frameworklooppathpuzzle} introduced a concise framework, called the \textit{T-metacell}, to prove the NP-hardness of such puzzles based on grid graph Hamiltonicity.

A T-metacell serves as a gadget representing a degree-3 vertex in a grid graph. Typically constructed as a square tile, it has exits on three of its four sides, allowing the loop to traverse the gadget between any pair of these exits. The gadget must be rotatable, and the loop can move between adjacent T-metacells only when both have exits on their shared edge. Moreover, the loop is required to visit every T-metacell exactly once.

Tang~\cite{tang2022frameworklooppathpuzzle} provided the following simple yet powerful argument showing that, for the Hamiltonicity problem in max-degree-3 spanning subgraph of a rectangular grid graph, all vertices in the graph -- including those with degree less than three -- can be represented by T-metacells.

Let $G$ be a max-degree-3 spanning subgraph of a rectangular grid graph, and let $H$ be its complementary graph, i.e. two vertices are adjacent in $H$ if and only if they are not adjacent in $G$. Because every vertex in $G$ has degree at least two, $H$ consists only of degree-1 and degree-2 vertices. We then assign a direction to the edges of $H$ such that each vertex has an indegree and outdegree of at most one. Consequently, we can construct $G$ using T-metacells by replacing each vertex with a T-metacell so that, for every degree-2 vertex, the third exit faces the direction of the outward-pointing edge of the corresponding vertex in $H$. This ensures that for every T-metacell representing a degree-2 vertex, the extra exit always faces a non-exit side of the adjacent cell, as illustrated in \cref{fig:t-cellconstruct}.

\begin{figure}[H]
  \centering
  \tikzset{
    gridstyle/.style={dashed, gray!60},
    boundarystyle/.style={black, thick},
    pathstyle/.style={green!60!black, thick}
  }
  \begin{subfigure}[t]{0.32\textwidth}
    \centering
    \begin{tikzpicture}[scale=1.65]
      \draw[gridstyle] (0,0) grid (2,3);
      \foreach \i in {(0,0),(0,1),(0,2),(1,0),(1,1),(2,0),(2,2)} {
        \draw[black, thick] \i -- ++(0,1);
      }
      \foreach \i in {(0,0),(0,3),(1,0),(1,1),(1,2),(1,3)} {
        \draw[black, thick] \i -- ++(1,0);
      }
      \filldraw[blue, thick] (0,0) circle (0.065 cm);
      \filldraw[blue, thick] (1,0) circle (0.065 cm);
      \filldraw[blue, thick] (2,0) circle (0.065 cm);
      \filldraw[blue, thick] (0,1) circle (0.065 cm);
      \filldraw[blue, thick] (1,1) circle (0.065 cm);
      \filldraw[blue, thick] (2,1) circle (0.065 cm);
      \filldraw[blue, thick] (0,2) circle (0.065 cm);
      \filldraw[blue, thick] (1,2) circle (0.065 cm);
      \filldraw[blue, thick] (2,2) circle (0.065 cm);
      \filldraw[blue, thick] (0,3) circle (0.065 cm);
      \filldraw[blue, thick] (1,3) circle (0.065 cm);
      \filldraw[blue, thick] (2,3) circle (0.065 cm);

    \end{tikzpicture}
    \caption{$G$}
  \end{subfigure}
  \begin{subfigure}[t]{0.32\textwidth}
    \centering
    \begin{tikzpicture}[scale=1]
      \foreach \i in {(2,3),(3,2)} {
        \draw[black, thick] \i -- ++(0,1);
      }
      \foreach \i in {(1,2),(1,3)} {
        \draw[black, thick] \i -- ++(1,0);
      }

      \draw[gridstyle] (0,0) grid (4,5);
      \filldraw[blue, thick] (1,1) circle (0.1cm);
      \filldraw[blue, thick] (2,1) circle (0.1cm);
      \filldraw[blue, thick] (3,1) circle (0.1cm);
      \filldraw[blue, thick] (1,2) circle (0.1cm);
      \filldraw[blue, thick] (2,2) circle (0.1cm);
      \filldraw[blue, thick] (3,2) circle (0.1cm);
      \filldraw[blue, thick] (1,3) circle (0.1cm);
      \filldraw[blue, thick] (2,3) circle (0.1cm);
      \filldraw[blue, thick] (3,3) circle (0.1cm);
      \filldraw[blue, thick] (1,4) circle (0.1cm);
      \filldraw[blue, thick] (2,4) circle (0.1cm);
      \filldraw[blue, thick] (3,4) circle (0.1cm);
      \filldraw[blue!20, thick] (1,0) circle (0.1cm);
      \filldraw[blue!20, thick] (2,0) circle (0.1cm);
      \filldraw[blue!20, thick] (3,0) circle (0.1cm);
      \filldraw[blue!20, thick] (1,5) circle (0.1cm);
      \filldraw[blue!20, thick] (2,5) circle (0.1cm);
      \filldraw[blue!20, thick] (3,5) circle (0.1cm);
      \filldraw[blue!20, thick] (0,1) circle (0.1cm);
      \filldraw[blue!20, thick] (0,2) circle (0.1cm);
      \filldraw[blue!20, thick] (0,3) circle (0.1cm);
      \filldraw[blue!20, thick] (0,4) circle (0.1cm);
      \filldraw[blue!20, thick] (4,1) circle (0.1cm);
      \filldraw[blue!20, thick] (4,2) circle (0.1cm);
      \filldraw[blue!20, thick] (4,3) circle (0.1cm);
      \filldraw[blue!20, thick] (4,4) circle (0.1cm);

      \foreach \i in {(0,3),(0,4),(1,3),(3,3)} {
        \draw[->, red, dashed, very thick] \i -- ++(1,0);
      }
      \foreach \i in {(1,1),(1,2),(2,2),(4,1),(4,2),(4,4)} {
        \draw[->, red, dashed, very thick] \i -- ++(-1,0);
      }
      \foreach \i in {(1,0),(1,4),(2,4),(3,4),(3,2),(2,3)} {
        \draw[->, red, dashed, very thick] \i -- ++(0,1);
      }
      \foreach \i in {(2,1),(3,1)} {
        \draw[->, red, dashed, very thick] \i -- ++(0,-1);
      }

    \end{tikzpicture}
    \caption{$H$}
  \end{subfigure}
  \begin{subfigure}[t]{0.32\textwidth}
    \centering
    \begin{tikzpicture}[scale=1.28]
      \draw[black, thick] (0,0) grid (3,4);

      \foreach \i in {(0.5,1.5),(2.5,0.5),(2.5,1.5),(2.5,3.5)} {
        \draw[red, thick] \i -- ++(0,0.5) \i -- ++(0,-0.5) \i -- ++(-0.5,0);
      }
      \foreach \i in {(0.5,0.5),(0.5,2.5),(0.5,3.5),(1.5,1.5),(1.5,2.5)} {
        \draw[red, thick] \i -- ++(0,0.5) \i -- ++(0.5,0) \i -- ++(0,-0.5);
      }
      \foreach \i in {(1.5,0.5),(2.5,2.5),(1.5,3.5)} {

        \draw[red, thick] \i -- ++(0,0.5) \i -- ++(0.5,0) \i -- ++(-0.5,0);
      }

    \end{tikzpicture}
    \caption{$G'$}
  \end{subfigure}

  \caption{Graphs $G$, $H$, and the corresponding graph $G'$ constructed from $G$ using T-metacells}
  \label{fig:t-cellconstruct}
\end{figure}

\subsection{Required-edge directed T-Metacell}


MIT Hardness Group \cite{37puzzles} proved that finding Hamiltonian cycles on a rectangular grid of required-edge directed T-metacells is ASP-complete by showing that the problem can solve the Hamiltonicity problem in a directed maximum-degree-3 spanning subgraph of a rectangular grid graph. A T-metacell is required-edge directed if each of its three exits has a forced direction, which the cycle must obey. When building a grid of T-metacells, the T-metacells need not be identical, as long as there exists a corresponding T-metacell configuration for every orientation,

\subsection{NP-Completeness and ASP-Completeness}
Building upon the T-metacell framework, MIT Hardness Group~\cite{37puzzles} further showed that solving the T-metacell problem is not only NP-complete but also ASP-complete.

A T-metacell gadget implies the ASP-Completeness of a puzzle if it satisfies all of the aforementioned properties and is \textit{locally unique} (i.e. for any given pair of exits, the local path within the gadget is uniquely determined). Even without this local uniqueness constraint, the framework remains sufficient to establish the NP-completeness of the puzzle.


Building on the methods of \cite{37puzzles}, which developed the existing T-metacell framework to be sufficient proof of NP-completeness, we prove that the puzzle \textit{All or Nothing}, \textit{Water Walk}, and \textit{Remembered Length} is NP-complete.

\subsection{Our Contribution}
In this paper, we prove the NP-completeness of three pencil puzzles, \textit{All or Nothing}, \textit{Water Walk}, and \textit{Remembered Length} with its generalized version, using the T-metacell framework. Specifically, we construct reductions from the problem of finding a Hamiltonian cycle in a maximum-degree-3 spanning subgraph of a rectangular grid graph.

As the constructed gadgets do not locally enforce a unique path between each pair of exits, our results establish NP-completeness but not ASP-completeness for these puzzles.

\section{NP-Completeness Proofs}
In this section, we present reductions demonstrating that the three puzzles are NP-complete with the T-metacell framework. In each reduction, vertices of the grid graph are mapped to T-metacell gadgets that simulate the constraints of the respective puzzles.

\subsection{All or Nothing}
All or Nothing is a pencil puzzle invented by Inaba Naoki \cite{nothing}. In this puzzle, a rectangular grid divided into polyominoes called \textit{regions} is given. The goal is to draw a non-crossing loop passing through orthogonally adjacent cells to satisfy the following three constraints.

\begin{enumerate}
  \item If the loop visits a region, it must pass through all cells in that region.
  \item The loop can enter and exit each region at most once.
  \item Two unvisited regions cannot be orthogonally adjacent to each other. See \cref{fig:nothing}.
\end{enumerate}

\begin{figure}[H]
  \centering
  \tikzset{
    gridstyle/.style={dashed, gray!60},
    boundarystyle/.style={black, thick},
    pathstyle/.style={blue, thick}
  }
  \newcommand{\drawcommonboundaries}{
    \begin{scope}[boundarystyle]
      \draw (0,0) rectangle (5,5);
      \draw (0,2) -- (5,2);
      \draw (1,3) -- (3,3);
      \draw (1,3) -- (1,5);
      \draw (2,2) -- (2,3);
      \draw (3,2) -- (3,3);
      \draw (4,3) -- (5,3);
      \draw (4,3) -- (4,5);
      \draw (2,1) -- (2,2);
      \draw (3,0) -- (3,1);
      \draw (2,1) -- (3,1);
    \end{scope}
  }
  \begin{tikzpicture}[scale=1]
    \draw[gridstyle] (0,0) grid (5,5);
    \drawcommonboundaries
  \end{tikzpicture}
  \hspace{2cm}
  \begin{tikzpicture}[scale=1]
    \filldraw [lightgray!40] (0,0) rectangle (2,2);
    \filldraw [lightgray!40] (2,0) rectangle (3,1);
    \filldraw [lightgray!40] (4,3) rectangle (5,5);
    \draw[gridstyle] (0,0) grid (5,5);
    \drawcommonboundaries
    \draw[pathstyle] (2.5, 2.5) -- (0.5, 2.5) -- (0.5, 4.5)
    -- (1.5, 4.5) -- (1.5, 3.5) -- (2.5, 3.5)
    -- (2.5, 4.5) -- (3.5, 4.5) -- (3.5, 2.5)
    -- (4.5, 2.5) -- (4.5, 0.5) -- (3.5, 0.5)
    -- (3.5, 1.5) -- (2.5, 1.5) -- (2.5, 2.5);
  \end{tikzpicture}

  \caption{An All or Nothing instance (left) and one of its solutions (right)}
  \label{fig:nothing}
\end{figure}

We use the term \emph{part} to denote an area of orthogonally connected cells in a T-metacell (which may not be fully enclosed by region boundary), and \emph{region} to denote a completely enclosed part in a full puzzle instance. A part or region is said to be \emph{leaf-rich} if it contains at least three \emph{leaves} -- cells adjacent to exactly one other cell within the same part or region. A region is said to be \emph{dead} if no valid loop can visit it.

We construct a T-metacell for \textit{All or Nothing} using an $11 \times 11$ grid gadget, as shown in \cref{fig:notone}. Each gadget consists of a big white region (which must be visited by the loop), a one-cell region enclosed within the big white region, and three partially enclosed parts, referred to as the \emph{foundational parts}. Note that all foundational parts are leaf-rich, and the gadget itself is not enclosed by region boundary.

\begin{figure}[H]
  \centering
  \begin{tikzpicture}[scale=0.5]
    \fill[red!20]
    (0,4) -- (1,4) -- (1,1) -- (3,1) -- (3,4) -- (5,4) -- (5,3) -- (7,3) -- (7,4) -- (8,4) -- (8,1) -- (10,1) -- (10,5) -- (11,5) -- (11,0) -- (0,0) -- cycle;
    \fill[red!20]
    (0,6) -- (1,6) -- (1,10) -- (4,10) -- (4,11) -- (0,11) -- cycle;
    \fill[red!20]
    (1,7) -- (2,7) -- (2,8) -- (1,8) -- cycle;
    \fill[red!20]
    (6,11) -- (6,10) -- (10,10) -- (10,7) -- (11,7) -- (11,11) -- cycle;
    \fill[red!20]
    (7,10) -- (7,9) -- (8,9) -- (8,10) -- cycle;
    \fill[red!20]
    (5,6) -- (6,6) -- (6,7) -- (5,7) -- cycle;
    \draw[step=1cm,gray] (0,0) grid (11,11);
    \draw[red, thick]
    (0,4) -- (1,4) -- (1,1) -- (3,1) -- (3,4) -- (5,4) -- (5,3) -- (7,3) -- (7,4) -- (8,4) -- (8,1) -- (10,1) -- (10,5) -- (11,5);
    \draw[red, thick] (0,6) -- (1,6) -- (1,7);
    \draw[red, thick] (1,8) -- (1,10) -- (4,10) -- (4,11);
    \draw[red, thick] (1,7) -- (2,7) -- (2,8) -- (1,8);
    \draw[red, thick] (6,11) -- (6,10) -- (7, 10);
    \draw[red, thick] (8, 10) -- (10,10) -- (10,7) -- (11,7);
    \draw[red, thick] (7,10) -- (7,9) -- (8,9) -- (8,10);
    \draw[red, thick] (5,6) -- (6,6) -- (6,7) -- (5,7) -- cycle;
    \draw[red, thick] (0,4) -- (0,6) (11,5) -- (11,7) (4,11) -- (6,11);
  \end{tikzpicture}
  \caption{A T-metacell gadget for All or Nothing, with the red lines indicating region boundaries, and the shaded area indicating foundational parts and the one-cell region}
  \label{fig:notone}
\end{figure}

When the full puzzle instance is constructed by combining all T-metacells, all parts will together form enclosed regions. The outer boundary of the puzzle is also treated as being enclosed by region boundary.

As shown in \cref{fig:one}, there exists a path between any pair of exits of the T-metacell that traverses the entire big white region.

\begin{figure}[H]
  \centering
  \begin{subfigure}[t]{0.32\textwidth}
    \centering
    \begin{tikzpicture}[scale=0.5]
      \fill[red!20]
      (0,4) -- (1,4) -- (1,1) -- (3,1) -- (3,4) -- (5,4) -- (5,3) -- (7,3) -- (7,4) -- (8,4) -- (8,1) -- (10,1) -- (10,5) -- (11,5) -- (11,0) -- (0,0) -- cycle;
      \fill[red!20]
      (0,6) -- (1,6) -- (1,10) -- (4,10) -- (4,11) -- (0,11) -- cycle;
      \fill[red!20]
      (1,7) -- (2,7) -- (2,8) -- (1,8) -- cycle;
      \fill[red!20]
      (6,11) -- (6,10) -- (10,10) -- (10,7) -- (11,7) -- (11,11) -- cycle;
      \fill[red!20]
      (7,10) -- (7,9) -- (8,9) -- (8,10) -- cycle;
      \fill[red!20]
      (5,6) -- (6,6) -- (6,7) -- (5,7) -- cycle;
      \draw[step=1cm,gray] (0,0) grid (11,11);
      \draw[red, thick]
      (0,4) -- (1,4) -- (1,1) -- (3,1) -- (3,4) -- (5,4) -- (5,3) -- (7,3) -- (7,4) -- (8,4) -- (8,1) -- (10,1) -- (10,5) -- (11,5);
      \draw[red, thick] (0,6) -- (1,6) -- (1,7);
      \draw[red, thick] (1,8) -- (1,10) -- (4,10) -- (4,11);
      \draw[red, thick] (1,7) -- (2,7) -- (2,8) -- (1,8);
      \draw[red, thick] (6,11) -- (6,10) -- (7, 10);
      \draw[red, thick] (8, 10) -- (10,10) -- (10,7) -- (11,7);
      \draw[red, thick] (7,10) -- (7,9) -- (8,9) -- (8,10);
      \draw[red, thick] (5,6) -- (6,6) -- (6,7) -- (5,7) -- cycle;
      \draw[blue, thick] (3.5,6.5) -- (3.5,5.5);
      \draw[blue, thick] (4.5,5.5) -- (4.5,6.5);
      \draw[blue, thick]
      (0,5.5) -- (0.5,5.5) -- (0.5,4.5) -- (1.5,4.5) -- (1.5,1.5) -- (2.5,1.5) -- (2.5,5.5) -- (1.5,5.5) -- (1.5,6.5) -- (2.5,6.5) -- (2.5,8.5) -- (1.5,8.5) -- (1.5,9.5) -- (3.5,9.5) -- (3.5,6.5);
      \draw[blue, thick]
      (3.5,5.5) -- (3.5,4.5) -- (5.5,4.5) -- (5.5,3.5) -- (6.5,3.5) -- (6.5,4.5) -- (8.5,4.5) -- (8.5,1.5) -- (9.5,1.5) -- (9.5,5.5) -- (10.5,5.5) -- (10.5,6.5) -- (8.5,6.5) -- (8.5,5.5) -- (4.5,5.5);
      \draw[blue, thick]
      (4.5,6.5) -- (4.5,8.5) -- (5.5,8.5) -- (5.5,7.5) -- (6.5,7.5) -- (6.5,6.5) -- (7.5,6.5) -- (7.5,7.5) -- (9.5,7.5) -- (9.5,9.5) -- (8.5,9.5) -- (8.5,8.5) -- (6.5,8.5) -- (6.5,9.5) -- (4.5,9.5) -- (4.5,10.5) -- (5.5,10.5) -- (5.5,11);
      \draw[red, thick] (0,4) -- (0,6) (11,5) -- (11,7) (4,11) -- (6,11);
    \end{tikzpicture}
    \caption{}
  \end{subfigure}
  \hspace{.1\textwidth}
  \begin{subfigure}[t]{0.32\textwidth}
    \centering
    \begin{tikzpicture}[scale=0.5]
      \fill[red!20]
      (0,4) -- (1,4) -- (1,1) -- (3,1) -- (3,4) -- (5,4) -- (5,3) -- (7,3) -- (7,4) -- (8,4) -- (8,1) -- (10,1) -- (10,5) -- (11,5) -- (11,0) -- (0,0) -- cycle;
      \fill[red!20]
      (0,6) -- (1,6) -- (1,10) -- (4,10) -- (4,11) -- (0,11) -- cycle;
      \fill[red!20]
      (1,7) -- (2,7) -- (2,8) -- (1,8) -- cycle;
      \fill[red!20]
      (6,11) -- (6,10) -- (10,10) -- (10,7) -- (11,7) -- (11,11) -- cycle;
      \fill[red!20]
      (7,10) -- (7,9) -- (8,9) -- (8,10) -- cycle;
      \fill[red!20]
      (5,6) -- (6,6) -- (6,7) -- (5,7) -- cycle;

      \draw[step=1cm,gray] (0,0) grid (11,11);
      \draw[red, thick]
      (0,4) -- (1,4) -- (1,1) -- (3,1) -- (3,4) -- (5,4) -- (5,3) -- (7,3) -- (7,4) -- (8,4) -- (8,1) -- (10,1) -- (10,5) -- (11,5);
      \draw[red, thick]
      (0,6) -- (1,6) -- (1,7);
      \draw[red, thick]
      (1,8) -- (1,10) -- (4,10) -- (4,11);
      \draw[red, thick]
      (1,7) -- (2,7) -- (2,8) -- (1,8);
      \draw[red, thick]
      (6,11) -- (6,10) -- (7, 10);
      \draw[red, thick]
      (8, 10) -- (10,10) -- (10,7) -- (11,7);
      \draw[red, thick]
      (7,10) -- (7,9) -- (8,9) -- (8,10);
      \draw[red, thick]
      (5,6) -- (6,6) -- (6,7) -- (5,7) -- cycle;
      \draw[blue, thick]
      (2.5,5.5) -- (3.5,5.5);
      \draw[blue, thick]
      (3.5,6.5) -- (4.5,6.5);

      \draw[blue, thick]
      (2.5,5.5) -- (2.5,1.5) -- (1.5,1.5) -- (1.5,4.5) -- (0.5,4.5) -- (0.5,5.5) -- (1.5,5.5) -- (1.5,6.5) -- (2.5,6.5) -- (2.5,8.5) -- (1.5,8.5) -- (1.5,9.5) -- (3.5,9.5) -- (3.5,6.5);
      \draw[blue, thick]
      (3.5,5.5) -- (3.5,4.5) -- (4.5,4.5) -- (4.5,5.5) -- (5.5,5.5) -- (5.5,3.5) -- (6.5,3.5) -- (6.5,4.5) -- (8.5,4.5) -- (8.5,1.5) -- (9.5,1.5) -- (9.5,6.5) -- (10.5,6.5) -- (10.5,5.5) -- (11,5.5);
      \draw[blue, thick]
      (4.5,6.5) -- (4.5,8.5) -- (5.5,8.5) -- (5.5,7.5) -- (7.5,7.5) -- (7.5,6.5) -- (6.5,6.5) -- (6.5,5.5) -- (8.5,5.5) -- (8.5,7.5) -- (9.5,7.5) -- (9.5,9.5) -- (8.5,9.5) -- (8.5,8.5) -- (6.5,8.5) -- (6.5,9.5) -- (4.5,9.5) -- (4.5,10.5) -- (5.5,10.5) -- (5.5,11);
      \draw[red, thick] (0,4) -- (0,6) (11,5) -- (11,7) (4,11) -- (6,11);
    \end{tikzpicture}
    \caption{}
  \end{subfigure}
  \begin{subfigure}[t]{0.32\textwidth}
    \centering
    \begin{tikzpicture}[scale=0.5]
      \fill[red!20]
      (0,4) -- (1,4) -- (1,1) -- (3,1) -- (3,4) -- (5,4) -- (5,3) -- (7,3) -- (7,4) -- (8,4) -- (8,1) -- (10,1) -- (10,5) -- (11,5) -- (11,0) -- (0,0) -- cycle;
      \fill[red!20]
      (0,6) -- (1,6) -- (1,10) -- (4,10) -- (4,11) -- (0,11) -- cycle;
      \fill[red!20]
      (1,7) -- (2,7) -- (2,8) -- (1,8) -- cycle;
      \fill[red!20]
      (6,11) -- (6,10) -- (10,10) -- (10,7) -- (11,7) -- (11,11) -- cycle;
      \fill[red!20]
      (7,10) -- (7,9) -- (8,9) -- (8,10) -- cycle;
      \fill[red!20]
      (5,6) -- (6,6) -- (6,7) -- (5,7) -- cycle;

      \draw[step=1cm,gray] (0,0) grid (11,11);

      \draw[red, thick]
      (0,4) -- (1,4) -- (1,1) -- (3,1) -- (3,4) -- (5,4) -- (5,3) -- (7,3) -- (7,4) -- (8,4) -- (8,1) -- (10,1) -- (10,5) -- (11,5);
      \draw[red, thick]
      (0,6) -- (1,6) -- (1,7);
      \draw[red, thick]
      (1,8) -- (1,10) -- (4,10) -- (4,11);
      \draw[red, thick]
      (1,7) -- (2,7) -- (2,8) -- (1,8);
      \draw[red, thick]
      (6,11) -- (6,10) -- (7, 10);
      \draw[red, thick]
      (8, 10) -- (10,10) -- (10,7) -- (11,7);
      \draw[red, thick]
      (7,10) -- (7,9) -- (8,9) -- (8,10);
      \draw[red, thick]
      (5,6) -- (6,6) -- (6,7) -- (5,7) -- cycle;

      \draw[blue, thick]
      (4.5,7.5) -- (5.5,7.5);
      \draw[blue, thick]
      (5.5,5.5) -- (6.5,5.5);

      \draw[blue, thick]
      (0,5.5) -- (0.5,5.5) -- (0.5,4.5) -- (1.5,4.5) -- (1.5,1.5) -- (2.5,1.5) -- (2.5,5.5) -- (1.5,5.5) -- (1.5,6.5) -- (2.5,6.5) -- (2.5,8.5) -- (1.5,8.5) -- (1.5,9.5) -- (3.5,9.5) -- (3.5,6.5) -- (4.5,6.5) -- (4.5,7.5);
      \draw[blue, thick]
      (5.5,7.5) -- (5.5,8.5) -- (4.5,8.5) -- (4.5,10.5) -- (5.5,10.5) -- (5.5,10.5) -- (5.5,9.5) -- (6.5,9.5) -- (6.5,7.5) -- (7.5,7.5) -- (7.5,8.5) -- (8.5,8.5) -- (8.5,9.5) -- (9.5,9.5) -- (9.5,7.5) -- (8.5,7.5) -- (8.5,5.5) -- (7.5,5.5) -- (7.5,6.5) -- (6.5,6.5) -- (6.5,5.5);
      \draw[blue, thick]
      (5.5,5.5) -- (3.5,5.5) -- (3.5,4.5) -- (5.5,4.5) -- (5.5,3.5) -- (6.5,3.5) -- (6.5,4.5) -- (8.5,4.5) -- (8.5,1.5) -- (9.5,1.5) -- (9.5,6.5) -- (10.5,6.5) -- (10.5,5.5) -- (11,5.5);
      \draw[red, thick] (0,4) -- (0,6) (11,5) -- (11,7) (4,11) -- (6,11);
    \end{tikzpicture}
    \caption{}
  \end{subfigure}
  \caption{Path between each pair of exits shown in blue}
  \label{fig:one}
\end{figure}

To show that the gadget satisfies all remaining properties of a T-metacell, we will prove that the three foundational parts cannot be visited by the loop.

\begin{lemma}
  The one-cell region inside every T-metacell is dead.
\end{lemma}

\begin{proof}
  A valid loop must cross the boundary of any visited region exactly twice. However, if the loop were to visit the one-cell region, it would necessarily cross the boundary of the big white region surrounding it twice. In addition, the loop must also cross the outer boundary of the big white region at least twice, resulting in at least four boundary crossings in total. Therefore, such a region cannot be visited.
\end{proof}

\begin{lemma}
  Any leaf-rich region is dead.
\end{lemma}

\begin{proof}
  A valid loop must cross the boundary of any visited region exactly twice. However, if the loop were to visit every cell of a region containing at least three leaves, it would necessarily cross the boundary of that region at least three times. Therefore, such a region cannot be visited.
\end{proof}

\begin{lemma}
  In the full puzzle instance, any region containing a foundational part of any T-metacell is leaf-rich.
\end{lemma}

\begin{proof}
  Consider any foundational part of any T-metacell. Within this part, there is one $+$-type cell (which is always a leaf) and two $?$-type cells located on its borders, as shown in \cref{fig:skibid}. Each $?$-type cell is a leaf when the T-metacell is considered in isolation, but may cease to be one when connected to an adjacent T-metacell, as shown in \cref{fig:skibid2}. For each $?$-type cell, the boundary on its side either coincides with the outer boundary of the full puzzle instance or connects to a foundational part of another T-metacell.

  In the former case, the $?$-type cell remains a leaf. In the latter case, this foundational part becomes connected to the foundational part of the adjacent T-metacell, which contains one $+$-type cell. Since the two $?$-type cells lie on different sides of the T-metacell, the corresponding $+$-type cells from the adjacent T-metacells must be different. Therefore, the region containing this part contains at least three leaves and thus is leaf-rich.
\end{proof}

\begin{figure}[H]
  \centering
  \begin{tikzpicture}[scale=0.5]
    \fill[red!20]
    (0,4) -- (1,4) -- (1,1) -- (3,1) -- (3,4) -- (5,4) -- (5,3) -- (7,3) -- (7,4) -- (8,4) -- (8,1) -- (10,1) -- (10,5) -- (11,5) -- (11,0) -- (0,0) -- cycle;
    \fill[red!20]
    (0,6) -- (1,6) -- (1,10) -- (4,10) -- (4,11) -- (0,11) -- cycle;
    \fill[red!20]
    (1,7) -- (2,7) -- (2,8) -- (1,8) -- cycle;
    \fill[red!20]
    (6,11) -- (6,10) -- (10,10) -- (10,7) -- (11,7) -- (11,11) -- cycle;
    \fill[red!20]
    (7,10) -- (7,9) -- (8,9) -- (8,10) -- cycle;
    \fill[red!20]
    (5,6) -- (6,6) -- (6,7) -- (5,7) -- cycle;
    \draw[step=1cm,gray] (0,0) grid (11,11);
    \draw[red, thick]
    (0,4) -- (1,4) -- (1,1) -- (3,1) -- (3,4) -- (5,4) -- (5,3) -- (7,3) -- (7,4) -- (8,4) -- (8,1) -- (10,1) -- (10,5) -- (11,5);
    \draw[red, thick] (0,6) -- (1,6) -- (1,7);
    \draw[red, thick] (1,8) -- (1,10) -- (4,10) -- (4,11);
    \draw[red, thick] (1,7) -- (2,7) -- (2,8) -- (1,8);
    \draw[red, thick] (6,11) -- (6,10) -- (7, 10);
    \draw[red, thick] (8, 10) -- (10,10) -- (10,7) -- (11,7);
    \draw[red, thick] (7,10) -- (7,9) -- (8,9) -- (8,10);
    \draw[red, thick] (0,4) -- (0,6);
    \draw[red, thick] (11,5) -- (11,7);
    \draw[red, thick] (4,11) -- (6,11);
    \draw[red, thick] (5,6) -- (6,6) -- (6,7) -- (5,7) -- cycle;
    \node at (0.5, 6.5) {\large $?$};
    \node at (0.5, 3.5) {\large $?$};
    \node at (10.5, 4.5) {\large $?$};
    \node at (10.5, 7.5) {\large $?$};
    \node at (6.5, 10.5) {\large $?$};
    \node at (3.5, 10.5) {\large $?$};
    \node at (1.5, 7.5) {\large $+$};
    \node at (7.5, 9.5) {\large $+$};
    \node at (7.5, 3.5) {\large $+$};
  \end{tikzpicture}
  \caption{$+$-type cells and $?$-type cells in the foundational parts of a T-metacell}
  \label{fig:skibid}
\end{figure}

\begin{figure}[H]
  \centering
  \begin{tikzpicture}[scale=0.5]
    \fill[red!20]
    (6,11) -- (6,10) -- (10,10) -- (10,7) -- (11,7) -- (11,11) -- cycle;
    \fill[red!20]
    (7,10) -- (7,9) -- (8,9) -- (8,10) -- cycle;
    \draw[step=1cm,gray] (0,0) grid (22,11);
    \draw[red, thick]
    (0,4) -- (1,4) -- (1,1) -- (3,1) -- (3,4) -- (5,4) -- (5,3) -- (7,3) -- (7,4) -- (8,4) -- (8,1) -- (10,1) -- (10,5) -- (11,5);

    \draw[red, thick] (0,6) -- (1,6) -- (1,7);
    \draw[red, thick] (1,8) -- (1,10) -- (4,10) -- (4,11);
    \draw[red, thick] (1,7) -- (2,7) -- (2,8) -- (1,8);
    \draw[red, thick] (6,11) -- (6,10) -- (7, 10);
    \draw[red, thick] (8, 10) -- (10,10) -- (10,7) -- (11,7);
    \draw[red, thick] (7,10) -- (7,9) -- (8,9) -- (8,10);
    \draw[red, thick] (11,6) -- (11,7);
    \draw[red, thick] (0,4) -- (0,6);
    \draw[red, thick] (11,5) -- (11,7);
    \draw[red, thick] (4,11) -- (6,11);
    \draw[red, thick] (11,4) -- (11,6);
    \draw[red, thick] (22,5) -- (22,7);
    \draw[red, thick] (15,11) -- (17,11);
    \draw[red, thick] (5,6) -- (6,6) -- (6,7) -- (5,7) -- cycle;
    \draw[black, thick] (-2,11) -- (22,11);
    \draw[black, thick] (22,-2) -- (22,11);
    \draw[red, thick]
    (11,4) -- (12,4) -- (12,1) -- (14,1) -- (14,4) -- (16,4) -- (16,3) -- (18,3) -- (18,4) -- (19,4) -- (19,1) -- (21,1) -- (21,5) -- (22,5);
    \draw[red, thick] (11,6) -- (12,6) -- (12,7);
    \draw[red, thick] (12,8) -- (12,10) -- (15,10) -- (15,11);
    \draw[red, thick] (12,7) -- (13,7) -- (13,8) -- (12,8);
    \draw[red, thick] (17,11) -- (17,10) -- (18,10);
    \draw[red, thick] (19,10) -- (21,10) -- (21,7) -- (22,7);
    \draw[red, thick] (18,10) -- (18,9) -- (19,9) -- (19,10);
    \draw[red, thick] (16,6) -- (17,6) -- (17,7) -- (16,7) -- cycle;
    \draw[black, thick] (9,11) -- (22,11);

    \node at (0.5, 6.5) {\large $?$};
    \node at (0.5, 3.5) {\large $?$};
    \node at (10.5, 4.5) {\large $?$};
    \node at (10.5, 7.5) {\large $?$};
    \node at (6.5, 10.5) {\large $?$};
    \node at (3.5, 10.5) {\large $?$};
    \node at (1.5, 7.5) {\large $+$};
    \node at (7.5, 9.5) {\large $+$};
    \node at (7.5, 3.5) {\large $+$};

    \node at (11.5, 6.5) {\large $?$};
    \node at (11.5, 3.5) {\large $?$};
    \node at (21.5, 4.5) {\large $?$};
    \node at (21.5, 7.5) {\large $?$};
    \node at (17.5, 10.5) {\large $?$};
    \node at (14.5, 10.5) {\large $?$};
    \node at (12.5, 7.5) {\large $+$};
    \node at (18.5, 9.5) {\large $+$};
    \node at (18.5, 3.5) {\large $+$};
  \end{tikzpicture}
  \caption{Two adjacent T-metacells at the top-right corner of the full puzzle instance, with the shaded area representing the foundational part of the left T-metacell. Its leftmost $?$-type cell remains a leaf as it touches the outer boundary, while its rightmost $?$-type cell ceases to be a leaf as it connects with another T-metacell.}
  \label{fig:skibid2}
\end{figure}

Hence, we can conclude that the gadget satisfies all properties of a T-metacell.

\subsection{Water Walk}
Water Walk is a pencil puzzle invented by Martin Ender \cite{water}. In this puzzle, a rectangular grid consisting of \textit{ground} (white) and \textit{water} (blue) cells is given, with some ground cells containing a number. The goal is to draw a non-crossing loop passing through orthogonally adjacent cells to satisfy the following three constraints.
\begin{enumerate}
  \item The loop must pass through all numbered cells.
  \item For each numbered cell, the number of continuous ground cells in the loop in the section where it passes through that cell must be exactly equal to the number written on it.
  \item The loop cannot pass through three consecutive water cells. See \cref{fig:waterWalk}.
\end{enumerate}




\begin{figure}[H]
  \centering
  \begin{subfigure}[t]{0.32\textwidth}
    \centering
    \begin{tikzpicture}[scale=1]
      \foreach \i in {0,1,2,3,4} {
        \foreach \j in {0,1,2,3,4} {
          \fill[cyan!20] (\i,\j) rectangle ++(1,1);
        }
      }

      \draw[dashed] (0,0) grid (5,5);

      \draw[white, fill=white]
      (2,0) -- (3,0) -- (3,1) -- (4,1) -- (4,4) -- (2,4) -- (2,3) -- (3,3) -- (3,2) -- (1,2) -- (1,4) -- (0,4) -- (0,1) -- (2,1) -- cycle;

      \draw[dashed] (1,1) -- (1,2) (2,1) -- (2,2) (3,1) -- (3,2) (3,3) -- (3,4) (0,2) -- (1,2) (0,3) -- (1,3) (2,1) -- (3,1) (3,2) -- (4,2) (3,3) -- (4,3);

      \draw[very thick]
      (2,0) -- (3,0) -- (3,1) -- (4,1) -- (4,4) -- (2,4) -- (2,3) -- (3,3) -- (3,2) -- (1,2) -- (1,4) -- (0,4) -- (0,1) -- (2,1) -- cycle;


      \node at (3.5,1.5) {\fontsize{30}{36}\selectfont 3};
      \node at (0.5,1.5) {\fontsize{30}{36}\selectfont 2};
      \node at (3.5,2.5) {\fontsize{30}{36}\selectfont 1};

    \end{tikzpicture}
    \caption{}
  \end{subfigure}
  \hspace{2cm}
  \begin{subfigure}[t]{0.32\textwidth}
    \centering
    \begin{tikzpicture}[scale=1]
      \foreach \i in {0,1,2,3,4} {
        \foreach \j in {0,1,2,3,4} {
          \fill[cyan!20] (\i,\j) rectangle ++(1,1);
        }
      }

      \draw[dashed] (0,0) grid (5,5);

      \draw[white, fill=white]
      (2,0) -- (3,0) -- (3,1) -- (4,1) -- (4,4) -- (2,4) -- (2,3) -- (3,3) -- (3,2) -- (1,2) -- (1,4) -- (0,4) -- (0,1) -- (2,1) -- cycle;

      \draw[very thick]
      (2,0) -- (3,0) -- (3,1) -- (4,1) -- (4,4) -- (2,4) -- (2,3) -- (3,3) -- (3,2) -- (1,2) -- (1,4) -- (0,4) -- (0,1) -- (2,1) -- cycle;

      \draw[dashed] (1,1) -- (1,2) (2,1) -- (2,2) (3,1) -- (3,2) (3,3) -- (3,4) (0,2) -- (1,2) (0,3) -- (1,3) (2,1) -- (3,1) (3,2) -- (4,2) (3,3) -- (4,3);

      \draw[blue, thick] (0.5,0.5) -- (2.5,0.5) -- (2.5,1.5) -- (4.5,1.5) -- (4.5,2.5) -- (2.5,2.5) -- (2.5,3.5) -- (1.5,3.5) -- (1.5,2.5) -- (0.5,2.5) -- cycle;

      \node at (3.5,1.5) {\fontsize{30}{36}\selectfont 3};
      \node at (0.5,1.5) {\fontsize{30}{36}\selectfont 2};
      \node at (3.5,2.5) {\fontsize{30}{36}\selectfont 1};

    \end{tikzpicture}
    \caption{}
  \end{subfigure}
  \caption{A Water Walk instance (left) and one of its solutions (right)}
  \label{fig:waterWalk}
\end{figure}

There exists a simple T-metacell for the puzzle, shown in \cref{fig:waterT-cell}. The paths for all three exit choices are trivial to find. The paths, however, are not unique; there are two different paths between exits on adjacent sides, and three between exits on opposite sides. Exiting through the long side is not possible, as that would create three consecutive water cells in the loop.

\begin{figure}[H]{}
  \centering
  \begin{subfigure}[t]{0.2\textwidth}
    \centering
    \begin{tikzpicture}[scale=.8]
      \foreach \i in {0,1,2,3,4} {
        \foreach \j in {0,1,2,3,4} {
          \fill[cyan!20] (\i,\j) rectangle ++(1,1);
        }
      }
      \draw[dashed] (0,0) grid (5,5);
      \draw[white, fill=white]
      (2,1) -- (2,4) -- (3,4) -- (3,3) -- (4,3)
      -- (4,2) -- (3,2) -- (3,1) -- cycle;
      \draw[very thick]
      (2,1) -- (2,4) -- (3,4) -- (3,3) -- (4,3)
      -- (4,2) -- (3,2) -- (3,1) -- cycle;
      \draw[dashed] (2,2) -- (3,2) (2,3) -- (3,3) (3,2) -- (3,3);
      \node at (2.5,2.5) {\fontsize{30}{36}\selectfont 3};
    \end{tikzpicture}
    \caption{}
  \end{subfigure}
  \hspace{0.1\linewidth}
  \begin{subfigure}[t]{0.2\textwidth}
    \centering
    \begin{tikzpicture}[scale=.8]
      \foreach \i in {0,1,2,3,4} {
        \foreach \j in {0,1,2,3,4} {
          \fill[cyan!20] (\i,\j) rectangle ++(1,1);
        }
      }
      \draw[dashed] (0,0) grid (5,5);
      \draw[white, fill=white]
      (2,1) -- (2,4) -- (3,4) -- (3,3) -- (4,3)
      -- (4,2) -- (3,2) -- (3,1) -- cycle;

      \draw[very thick]
      (2,1) -- (2,4) -- (3,4) -- (3,3) -- (4,3)
      -- (4,2) -- (3,2) -- (3,1) -- cycle;
      \draw[dashed] (2,2) -- (3,2) (2,3) -- (3,3) (3,2) -- (3,3);
      \draw[blue, very thick] (2.5,0) -- (2.5,2.5);
      \draw[blue, very thick] (2.5,2.5) -- (5,2.5);
      \node at (2.5,2.5) {\fontsize{30}{36}\selectfont 3};
    \end{tikzpicture}
    \caption{}
  \end{subfigure}
  \hspace{0.1\linewidth}
  \begin{subfigure}[t]{0.20\textwidth}
    \centering
    \begin{tikzpicture}[scale=.8]
      \foreach \i in {0,1,2,3,4} {
        \foreach \j in {0,1,2,3,4} {
          \fill[cyan!20] (\i,\j) rectangle ++(1,1);
        }
      }

      \draw[dashed] (0,0) grid (5,5);

      \draw[white, fill=white]
      (2,1) -- (2,4) -- (3,4) -- (3,3) -- (4,3)
      -- (4,2) -- (3,2) -- (3,1) -- cycle;

      \draw[very thick]
      (2,1) -- (2,4) -- (3,4) -- (3,3) -- (4,3)
      -- (4,2) -- (3,2) -- (3,1) -- cycle;

      \draw[dashed] (2,2) -- (3,2) (2,3) -- (3,3) (3,2) -- (3,3);

      \draw[blue, very thick] (2.5,0) -- (2.5,3.5) -- (3.5,3.5) -- (3.5, 2.5) -- (5, 2.5);

      \node at (2.5,2.5) {\fontsize{30}{36}\selectfont 3};

    \end{tikzpicture}
    \caption{}
  \end{subfigure}

  \begin{subfigure}[t]{0.20\textwidth}
    \centering
    \begin{tikzpicture}[scale=.8]
      \foreach \i in {0,1,2,3,4} {
        \foreach \j in {0,1,2,3,4} {
          \fill[cyan!20] (\i,\j) rectangle ++(1,1);
        }
      }
      \draw[dashed] (0,0) grid (5,5);
      \draw[white, fill=white]
      (2,1) -- (2,4) -- (3,4) -- (3,3) -- (4,3)
      -- (4,2) -- (3,2) -- (3,1) -- cycle;
      \draw[very thick]
      (2,1) -- (2,4) -- (3,4) -- (3,3) -- (4,3)
      -- (4,2) -- (3,2) -- (3,1) -- cycle;
      \draw[dashed] (2,2) -- (3,2) (2,3) -- (3,3) (3,2) -- (3,3);
      \draw[blue, very thick] (2.5,0) -- (2.5,5);
      \node at (2.5,2.5) {\fontsize{30}{36}\selectfont 3};
    \end{tikzpicture}
    \caption{}
  \end{subfigure}
  \hspace{0.1\linewidth}
  \begin{subfigure}[t]{0.20\textwidth}
    \centering
    \begin{tikzpicture}[scale=.8]
      \foreach \i in {0,1,2,3,4} {
        \foreach \j in {0,1,2,3,4} {
          \fill[cyan!20] (\i,\j) rectangle ++(1,1);
        }
      }
      \draw[dashed] (0,0) grid (5,5);
      \draw[white, fill=white]
      (2,1) -- (2,4) -- (3,4) -- (3,3) -- (4,3)
      -- (4,2) -- (3,2) -- (3,1) -- cycle;
      \draw[very thick]
      (2,1) -- (2,4) -- (3,4) -- (3,3) -- (4,3)
      -- (4,2) -- (3,2) -- (3,1) -- cycle;
      \draw[dashed] (2,2) -- (3,2) (2,3) -- (3,3) (3,2) -- (3,3);
      \draw[blue, very thick] (2.5,0) -- (2.5,1.5) -- (3.5,1.5) -- (3.5,2.5) -- (2.5,2.5) -- (2.5,5);
      \node at (2.5,2.5) {\fontsize{30}{36}\selectfont 3};
    \end{tikzpicture}
    \caption{}
  \end{subfigure}
  \hspace{0.1\linewidth}
  \begin{subfigure}[t]{0.20\textwidth}
    \centering
    \begin{tikzpicture}[scale=.8]
      \foreach \i in {0,1,2,3,4} {
        \foreach \j in {0,1,2,3,4} {
          \fill[cyan!20] (\i,\j) rectangle ++(1,1);
        }
      }
      \draw[dashed] (0,0) grid (5,5);
      \draw[white, fill=white]
      (2,1) -- (2,4) -- (3,4) -- (3,3) -- (4,3)
      -- (4,2) -- (3,2) -- (3,1) -- cycle;
      \draw[very thick]
      (2,1) -- (2,4) -- (3,4) -- (3,3) -- (4,3)
      -- (4,2) -- (3,2) -- (3,1) -- cycle;
      \draw[dashed] (2,2) -- (3,2) (2,3) -- (3,3) (3,2) -- (3,3);
      \draw[blue, very thick] (2.5,0) -- (2.5,2.5) -- (3.5,2.5) -- (3.5,3.5) -- (2.5,3.5) -- (2.5,5);
      \node at (2.5,2.5) {\fontsize{30}{36}\selectfont 3};
    \end{tikzpicture}
    \caption{}
  \end{subfigure}
  \caption{T-metacell for Water Walk (a), the two paths between exits on adjacent sides (b-c), and the three paths between exits on opposite sides (e-f)}
  \label{fig:waterT-cell}
\end{figure}

\subsection{Remembered Length}

Remembered Length is a pencil puzzle invented by Palmer Mebane \cite{remlen}. In this puzzle, a rectangular grid divided into polyominoes called \textit{regions}. with some cell shaded and some region labeled with a positive integer. The goal is to draw a non-crossing directed loop passing through orthogonally adjacent cells to satisfy the following constraints.
\begin{enumerate}
  \item The loop must pass through all unshaded cells.
  \item Each time the loop exits a region containing a number, its section in the next region must consist of exactly that number of cells. See \cref{fig:remlen1}.
\end{enumerate}

\begin{figure}[H]
  \centering
  \tikzset{
    gridstyle/.style={dashed, gray!60},
    boundarystyle/.style={black, thick},
    pathstyle/.style={blue, thick}
  }
  \newcommand{\drawcommonboundaries}{
    \begin{scope}[boundarystyle]
      \draw (0,0) rectangle (4,4);
      \draw (2,0) -- (2,1) -- (1,1) -- (1,2) -- (0,2);
      \draw (1,2) -- (2,2) -- (2,3) -- (1,3) -- (0,3);
      \draw (2,3) -- (3,3) -- (3,2) -- (3,1) -- (3,0);
      \draw (3,2) -- (4,2);
      \node at (0.5,2.5) {\fontsize{18}{20}\selectfont 3};
      \node at (0.5,3.5) {\fontsize{18}{20}\selectfont 1};
    \end{scope}
  }
  \begin{tikzpicture}[scale=1]
    \draw[gridstyle] (0,0) grid (4,4);
    \drawcommonboundaries
  \end{tikzpicture}
  \hspace{2cm}
  \begin{tikzpicture}[scale=1]
    \draw[gridstyle] (0,0) grid (4,4);
    \drawcommonboundaries
    \draw[pathstyle] (0.5, 0.5) -- (1.5, 0.5) -- (1.5, 1.5) -- (1.5, 2.5) -- (2.5, 2.5) -- (2.5,1.5) -- (2.5, 0.5) -- (3.5,0.5) -- (3.5,1.5) -- (3.5,2.5) -- (3.5,3.5) -- (2.5,3.5) -- (1.5, 3.5) -- (0.5, 3.5) -- (0.5, 2.5) -- (0.5, 1.5) -- (0.5, 0.5);
    \draw[pathstyle] (0.9, 0.3) -- (1.1, 0.5) -- (0.9, 0.7);
    \draw[pathstyle] (2.9, 0.3) -- (3.1, 0.5) -- (2.9, 0.7);
    \draw[pathstyle] (1.9, 2.3) -- (2.1, 2.5) -- (1.9, 2.7);
    \draw[pathstyle] (2.1, 3.3) -- (1.9, 3.5) -- (2.1, 3.7);
    \draw[pathstyle] (3.3, 1.9) -- (3.5, 2.1) -- (3.7, 1.9);
    \draw[pathstyle] (1.3, 1.4) -- (1.5, 1.6) -- (1.7, 1.4);
    \draw[pathstyle] (2.3, 1.6) -- (2.5, 1.4) -- (2.7, 1.6);
    \draw[pathstyle] (0.3, 2.1) -- (0.5, 1.9) -- (0.7, 2.1);
  \end{tikzpicture}

  \caption{A Remerbered Length instance (left) and one of its solutions (right)}
  \label{fig:remlen1}
\end{figure}

MIT Hardness Group \cite{37puzzles}  proved that finding a Hamiltonian path on required-edge subgraph of rectangular grid graph is NP-complete; thus Remembered Length is trivially NP-complete, and the reduction is by taking any bounding rectangular grid graph then shading any vertices not present in the original grid graph.

Furthermore, we prove that the generalized version---in which no cells are shaded---is also NP-complete.  We present, in \cref{fig:remlen}, four directed T-metacells which, rotated, cover all type of degree-3 vertex in directed grid graph.

We use the term \textit{ocean} to denote the biggest region on the T-metacell, and \textit{island} to denote any other regions. The island that is closer to the border of the ocean is said to be the \textit{in-island}, having minimum distance of two cells.  The other is called the \textit{out-island}, having minimum distance of three cells.

Each time the loop exits a T-metacell, that T-metacell updates the remembered length to be two.  Therefore, it can only enter the next T-metacell on the sides that are closest to an in-island then go straight to that island. As there is one region on the in-island that contains a number five less than the number of cell in the ocean, all but five cells of the ocean must be visited by the loop immediately after the loop visit this region. Thus, the loop can only enter T-metacell at the in-island (using two cells) and exit in a straight line at some side of the out-island (using three cells) once. Moreover, the loop can only enter and exit any island once, because there are not enough cells to re-enter an island.

To conclude, when entering the T-metacell, the loop need to enter the \textit{in-island} via the shortest path, visit all cells on the island. After exiting the \textit{in-island}, it must visit the cells in the \textit{ocean}, leaving three remaining unvisited.  Lastly, to leave, it will enter the \textit{out-island} and exit to the nearest \textit{ocean} border.

The loop always exit at the center of some side of the T-metacell, thus, it must enter at the center of some side of the T-metacell too.


\begin{figure}[htbp]
  \centering
  \begin{subfigure}[b]{\textwidth}
    \begin{tikzpicture}[scale=0.5]
      \draw[step=1] (0,0) grid (13,13);
      \draw[very thick] (0,0) rectangle (13,13);

      \draw[very thick] (3,6) rectangle (5,7);
      \draw[very thick] (6,2) rectangle (8,3);
      \draw[very thick] (6,3) rectangle (7,4);
      \draw[very thick] (6,4) rectangle (7,5);
      \draw[very thick] (6,5) rectangle (7,6);
      \draw[very thick] (6,6) rectangle (7,7);
      \draw[very thick] (6,7) rectangle (7,8);
      \draw[very thick] (7,3) rectangle (8,4);
      \draw[very thick] (7,4) rectangle (8,5);
      \draw[very thick] (7,5) rectangle (8,6);
      \draw[very thick] (7,6) rectangle (8,7);
      \draw[very thick] (7,7) rectangle (8,8);
      \draw[very thick] (8,6) rectangle (9,7);
      \draw[very thick] (8,7) rectangle (9,8);
      \draw[very thick] (9,6) rectangle (10,7);
      \draw[very thick] (9,7) rectangle (10,8);
      \draw[very thick] (10,6) rectangle (11,8);

      \tikzset{pathline/.style={blue, very thick, cap=round}}

      \tikzset{
        clearlabel/.style={text=black}
      }
      \node[clearlabel] at (3.5,6.5) {\fontsize{10}{12}\selectfont 3};
      \node[clearlabel] at (6.5,7.5) {\fontsize{8}{12}\selectfont 144};
      \node[clearlabel] at (0.5,12.5) {\fontsize{10}{12}\selectfont 2};

    \end{tikzpicture}
    \centering
    \subcaption{}
  \end{subfigure}

  \vspace{1em}

  \begin{subfigure}[b]{0.45\textwidth}
    \centering
    \begin{tikzpicture}[scale=0.5]
      \draw[step=1] (0,0) grid (13,13);
      \draw[very thick] (0,0) rectangle (13,13);

      \draw[very thick] (3,6) rectangle (5,7);
      \draw[very thick] (6,2) rectangle (8,3);
      \draw[very thick] (6,3) rectangle (7,4);
      \draw[very thick] (6,4) rectangle (7,5);
      \draw[very thick] (6,5) rectangle (7,6);
      \draw[very thick] (6,6) rectangle (7,7);
      \draw[very thick] (6,7) rectangle (7,8);
      \draw[very thick] (7,3) rectangle (8,4);
      \draw[very thick] (7,4) rectangle (8,5);
      \draw[very thick] (7,5) rectangle (8,6);
      \draw[very thick] (7,6) rectangle (8,7);
      \draw[very thick] (7,7) rectangle (8,8);
      \draw[very thick] (8,6) rectangle (9,7);
      \draw[very thick] (8,7) rectangle (9,8);
      \draw[very thick] (9,6) rectangle (10,7);
      \draw[very thick] (9,7) rectangle (10,8);
      \draw[very thick] (10,6) rectangle (11,8);

      \tikzset{pathline/.style={blue, very thick, cap=round}}
      \draw[pathline] (0.0,6.5) --(4.5,6.5) --(4.5,5.5) --(0.5,5.5) --(0.5,4.5) --(4.5,4.5) --(4.5,3.5) --(0.5,3.5) --(0.5,2.5) --(4.5,2.5) --(4.5,2.5) --(4.5,1.5) --(0.5,1.5) --(0.5,0.5) --(5.5,0.5) --(5.5,7.5) --(0.5,7.5) --(0.5,8.5) --(5.5,8.5) --(5.5,9.5) --(0.5,9.5) --(0.5,10.5) --(5.5,10.5) --(5.5,11.5) --(0.5,11.5) --(0.5,12.5) --(6.5,12.5) --(6.5,11.5) --(7.5,11.5) --(7.5,12.5) --(8.5,12.5) --(8.5,11.5) --(9.5,11.5) --(9.5,12.5) --(10.5,12.5) --(10.5,11.5) --(11.5,11.5) --(11.5,12.5) --(12.5,12.5) --(12.5,10.5) --(6.5,10.5) --(6.5,9.5) --(12.5,9.5) --(12.5,8.5) --(12.5,0.5) --(8.5,0.5) --(7.5,0.5) --(7.5,1.5) --(11.5,1.5) --(11.5,2.5) --(8.5,2.5) --(8.5,3.5) --(11.5,3.5) --(11.5,4.5) --(8.5,4.5) --(8.5,4.5) --(8.5,5.5) --(11.5,5.5) --(11.5,8.5) --(6.5,8.5) --(6.5,7.5) --(10.5,7.5) --(10.5,6.5) --(6.5,6.5) --(6.5,5.5) --(7.5,5.5) --(7.5,4.5) --(6.5,4.5) --(6.5,3.5) --(7.5,3.5) --(7.5,2.5) --(6.5,2.5) --(6.5,0.5) --(6.5,0.0);

      \tikzset{
        clearlabel/.style={font=\bfseries, text=black}
      }
      \node[clearlabel] at (3.25,6.75) {\fontsize{8}{12}\selectfont 3};
      \node[clearlabel] at (6.5,7.75) {\fontsize{7}{12}\selectfont 144};
      \node[clearlabel] at (0.25,12.75) {\fontsize{8}{12}\selectfont 2};

      
      \draw[pathline,->] (6.5,0)--(6.5,.8);
      \draw[pathline,->] (0.8,6.5)--(0,6.5);
    \end{tikzpicture}
    \subcaption{}
  \end{subfigure}\hfill
  \begin{subfigure}[b]{0.45\textwidth}
    \centering
    \begin{tikzpicture}[scale=0.5]
      \draw[step=1] (0,0) grid (13,13);
      \draw[very thick] (0,0) rectangle (13,13);

      \draw[very thick] (3,6) rectangle (5,7);
      \draw[very thick] (6,2) rectangle (8,3);
      \draw[very thick] (6,3) rectangle (7,4);
      \draw[very thick] (6,4) rectangle (7,5);
      \draw[very thick] (6,5) rectangle (7,6);
      \draw[very thick] (6,6) rectangle (7,7);
      \draw[very thick] (6,7) rectangle (7,8);
      \draw[very thick] (7,3) rectangle (8,4);
      \draw[very thick] (7,4) rectangle (8,5);
      \draw[very thick] (7,5) rectangle (8,6);
      \draw[very thick] (7,6) rectangle (8,7);
      \draw[very thick] (7,7) rectangle (8,8);
      \draw[very thick] (8,6) rectangle (9,7);
      \draw[very thick] (8,7) rectangle (9,8);
      \draw[very thick] (9,6) rectangle (10,7);
      \draw[very thick] (9,7) rectangle (10,8);
      \draw[very thick] (10,6) rectangle (11,8);

      \tikzset{pathline/.style={blue, very thick,cap=round}}
      \draw[pathline] (13.0,6.5) --(10.5,6.5) --(10.5,7.5) --(9.5,7.5) --(9.5,6.5) --(8.5,6.5) --(8.5,7.5) --(7.5,7.5) --(7.5,2.5) --(6.5,2.5) --(6.5,8.5) --(11.5,8.5) --(11.5,7.5) --(12.5,7.5) --(12.5,12.5) --(11.5,12.5) --(11.5,9.5) --(10.5,9.5) --(10.5,12.5) --(9.5,12.5) --(9.5,10.5) --(9.5,9.5) --(8.5,9.5) --(8.5,12.5) --(7.5,12.5) --(7.5,9.5) --(6.5,9.5) --(6.5,12.5) --(1.5,12.5) --(0.5,12.5) --(0.5,12.5) --(0.5,11.5) --(5.5,11.5) --(5.5,10.5) --(0.5,10.5) --(0.5,9.5) --(5.5,9.5) --(5.5,8.5) --(0.5,8.5) --(0.5,7.5) --(5.5,7.5) --(5.5,1.5) --(8.5,1.5) --(8.5,5.5) --(12.5,5.5) --(12.5,4.5) --(9.5,4.5) --(9.5,3.5) --(12.5,3.5) --(12.5,3.5) --(12.5,2.5) --(9.5,2.5) --(9.5,1.5) --(12.5,1.5) --(12.5,0.5) --(0.5,0.5) --(0.5,1.5) --(4.5,1.5) --(4.5,2.5) --(0.5,2.5) --(0.5,3.5) --(4.5,3.5) --(4.5,4.5) --(0.5,4.5) --(0.5,5.5) --(4.5,5.5) --(4.5,6.5) --(0.5,6.5);
      \tikzset{
        clearlabel/.style={font=\bfseries, text=black}
      }
      
      \node[clearlabel] at (3.25,6.75) {\fontsize{8}{12}\selectfont 3};
      \node[clearlabel] at (6.5,7.75) {\fontsize{7}{12}\selectfont 144};
      \node[clearlabel] at (0.25,12.75) {\fontsize{8}{12}\selectfont 2};
      \draw[pathline,->] (13,6.5)--(12.2,6.5);
      \draw[pathline,->] (0.8,6.5)--(0,6.5);
    \end{tikzpicture}
    \subcaption{}
  \end{subfigure}
  
  \caption{T-metacells for Remembered Length represent a vertex with 2 in-edges on adjacent sides (a), and the path for each chosen in-edge (b-c)}
  \label{fig:all}
\end{figure}

\begin{figure}[htbp]
  \ContinuedFloat
  \centering
  
  \begin{subfigure}[b]{\textwidth}
    \centering
    \begin{tikzpicture}[scale=0.5]
      \draw[step=1] (0,0) grid (13,13);
      \draw[very thick] (0,0) rectangle (13,13);

      \draw[very thick] (2,5) rectangle (3,7);
      \draw[very thick] (2,7) rectangle (3,8);
      \draw[very thick] (2,8) rectangle (3,9);
      \draw[very thick] (3,5) rectangle (4,6);
      \draw[very thick] (3,6) rectangle (4,7);
      \draw[very thick] (3,7) rectangle (4,8);
      \draw[very thick] (3,8) rectangle (4,9);
      \draw[very thick] (4,7) rectangle (5,8);
      \draw[very thick] (4,8) rectangle (5,9);
      \draw[very thick] (5,7) rectangle (6,8);
      \draw[very thick] (5,8) rectangle (6,9);
      \draw[very thick] (6,3) rectangle (7,5);
      \draw[very thick] (6,7) rectangle (7,8);
      \draw[very thick] (6,8) rectangle (7,9);
      \draw[very thick] (7,7) rectangle (8,8);
      \draw[very thick] (7,8) rectangle (8,9);
      \draw[very thick] (8,7) rectangle (9,8);
      \draw[very thick] (8,8) rectangle (9,9);
      \draw[very thick] (9,7) rectangle (10,8);
      \draw[very thick] (9,8) rectangle (10,9);
      \draw[very thick] (10,5) rectangle (11,7);
      \draw[very thick] (10,7) rectangle (11,8);
      \draw[very thick] (10,8) rectangle (11,9);
      \draw[very thick] (9,5) rectangle (10,6);
      \draw[very thick] (9,6) rectangle (10,7);

      \tikzset{
        clearlabel/.style={text=black}
      }
      \node[clearlabel] at (0.5,12.5) {\fontsize{10}{12}\selectfont 2};
      \node[clearlabel] at (2.5,6.5) {\fontsize{10}{12}\selectfont 1};
      \node[clearlabel] at (5.5,8.5) {\fontsize{8}{12}\selectfont 136};
      \node[clearlabel] at (6.5,4.5) {\fontsize{10}{12}\selectfont 3};
    \end{tikzpicture}
    
    \centering
    \subcaption{}
  \end{subfigure}

  \vspace{1em}

  \begin{subfigure}[b]{0.45\textwidth}
    \centering
    \begin{tikzpicture}[scale=0.5]
      \draw[step=1] (0,0) grid (13,13);
      \draw[very thick] (0,0) rectangle (13,13);

      \draw[very thick] (2,5) rectangle (3,7);
      \draw[very thick] (2,7) rectangle (3,8);
      \draw[very thick] (2,8) rectangle (3,9);
      \draw[very thick] (3,5) rectangle (4,6);
      \draw[very thick] (3,6) rectangle (4,7);
      \draw[very thick] (3,7) rectangle (4,8);
      \draw[very thick] (3,8) rectangle (4,9);
      \draw[very thick] (4,7) rectangle (5,8);
      \draw[very thick] (4,8) rectangle (5,9);
      \draw[very thick] (5,7) rectangle (6,8);
      \draw[very thick] (5,8) rectangle (6,9);
      \draw[very thick] (6,3) rectangle (7,5);
      \draw[very thick] (6,7) rectangle (7,8);
      \draw[very thick] (6,8) rectangle (7,9);
      \draw[very thick] (7,7) rectangle (8,8);
      \draw[very thick] (7,8) rectangle (8,9);
      \draw[very thick] (8,7) rectangle (9,8);
      \draw[very thick] (8,8) rectangle (9,9);
      \draw[very thick] (9,7) rectangle (10,8);
      \draw[very thick] (9,8) rectangle (10,9);
      \draw[very thick] (10,5) rectangle (11,7);
      \draw[very thick] (10,7) rectangle (11,8);
      \draw[very thick] (10,8) rectangle (11,9);
      \draw[very thick] (9,5) rectangle (10,6);
      \draw[very thick] (9,6) rectangle (10,7);

      \tikzset{pathline/.style={blue, very thick,cap=round}}
      \draw[pathline] (0.0,6.5) --(2.5,6.5) --(2.5,5.5) --(3.5,5.5) --(3.5,7.5) --(2.5,7.5) --(2.5,8.5) --(4.5,8.5);
      \draw[pathline] (4.5,8.5) --(4.5,7.5) --(9.5,7.5) --(9.5,5.5) --(10.5,5.5) --(10.5,8.5) --(5.5,8.5);
      \draw[pathline] (5.5,8.5) --(5.5,9.5) --(1.5,9.5) --(1.5,7.5) --(0.5,7.5) --(0.5,12.5) --(1.5,12.5) --(1.5,10.5) --(2.5,10.5) --(2.5,12.5) --(3.5,12.5) --(3.5,10.5) --(4.5,10.5) --(4.5,12.5) --(5.5,12.5)--(5.5,11.5);
      \draw[pathline] (5.5,11.5) --(5.5,10.5) --(6.5,10.5) --(6.5,9.5) --(7.5,9.5) --(7.5,11.5) --(6.5,11.5) --(6.5,12.5) --(8.5,12.5);
      \draw[pathline] (8.5,12.5) --(8.5,9.5) --(9.5,9.5) --(9.5,12.5) --(10.5,12.5) --(10.5,9.5) --(10.5,9.5) --(11.5,9.5) --(11.5,12.5) --(12.5,12.5) --(12.5,8.5)--(11.5,8.5);
      \draw[pathline] (11.5,8.5) --(11.5,7.5) --(12.5,7.5) --(12.5,7.5) --(12.5,6.5) --(11.5,6.5) --(11.5,5.5) --(12.5,5.5) --(12.5,4.5)--(11.5,4.5);
      \draw[pathline] (11.5,4.5) --(11.5,3.5) --(12.5,3.5) --(12.5,2.5) --(11.5,2.5) --(11.5,1.5) --(11.5,1.5) --(12.5,1.5) --(12.5,0.5) --(10.5,0.5) --(7.5,0.5)--(7.5,1.5);
      \draw[pathline] (7.5,1.5) --(10.5,1.5) --(10.5,2.5) --(7.5,2.5) --(7.5,3.5) --(10.5,3.5) --(10.5,4.5) --(7.5,4.5) --(7.5,5.5) --(8.5,5.5);
      \draw[pathline] (8.5,5.5) --(8.5,6.5) --(6.5,6.5) --(6.5,5.5) --(5.5,5.5) --(5.5,6.5) --(4.5,6.5) --(4.5,5.5) --(4.5,4.5) --(1.5,4.5) --(1.5,5.5)--(0.5,5.5);
      \draw[pathline] (0.5,5.5) --(0.5,0.5) --(1.5,0.5) --(1.5,3.5) --(2.5,3.5) --(2.5,0.5) --(3.5,0.5) --(3.5,3.5) --(4.5,3.5) --(4.5,0.5);
      \draw[pathline] (4.5,0.5) --(4.5,0.5) --(5.5,0.5) --(5.5,4.5) --(6.5,4.5) --(6.5,0.5) --(6.5,0.0);
      \draw[pathline,->] (6.5,.8)--(6.5,0);
      \draw[pathline,->] (0,6.5)--(0.8,6.5);

      \tikzset{
        clearlabel/.style={font=\bfseries, text=black}
      }
      \node[clearlabel] at (0.25,12.75) {\fontsize{8}{12}\selectfont 2};
      \node[clearlabel] at (2.25,6.75) {\fontsize{8}{12}\selectfont 1};
      \node[clearlabel] at (5.5,8.75) {\fontsize{7}{12}\selectfont 136};
      \node[clearlabel] at (6.25,4.75) {\fontsize{8}{12}\selectfont 3};
    \end{tikzpicture}

    \subcaption{}
    \label{fig:V}
  \end{subfigure}\hfill
  \begin{subfigure}[b]{0.45\textwidth}
    \centering
    \begin{tikzpicture}[scale=0.5]
      \draw[step=1] (0,0) grid (13,13);
      \draw[very thick] (0,0) rectangle (13,13);

      \draw[very thick] (2,5) rectangle (3,7);
      \draw[very thick] (2,7) rectangle (3,8);
      \draw[very thick] (2,8) rectangle (3,9);
      \draw[very thick] (3,5) rectangle (4,6);
      \draw[very thick] (3,6) rectangle (4,7);
      \draw[very thick] (3,7) rectangle (4,8);
      \draw[very thick] (3,8) rectangle (4,9);
      \draw[very thick] (4,7) rectangle (5,8);
      \draw[very thick] (4,8) rectangle (5,9);
      \draw[very thick] (5,7) rectangle (6,8);
      \draw[very thick] (5,8) rectangle (6,9);
      \draw[very thick] (6,3) rectangle (7,5);
      \draw[very thick] (6,7) rectangle (7,8);
      \draw[very thick] (6,8) rectangle (7,9);
      \draw[very thick] (7,7) rectangle (8,8);
      \draw[very thick] (7,8) rectangle (8,9);
      \draw[very thick] (8,7) rectangle (9,8);
      \draw[very thick] (8,8) rectangle (9,9);
      \draw[very thick] (9,5) rectangle (10,6);
      \draw[very thick] (9,6) rectangle (10,7);
      \draw[very thick] (9,7) rectangle (10,8);
      \draw[very thick] (9,8) rectangle (10,9);
      \draw[very thick] (10,5) rectangle (11,7);
      \draw[very thick] (10,7) rectangle (11,8);
      \draw[very thick] (10,8) rectangle (11,9);

      \tikzset{pathline/.style={blue, very thick,cap=round}}
      \draw[pathline] (13.0,6.5) --(10.5,6.5) --(10.5,5.5) --(9.5,5.5) --(9.5,7.5) --(10.5,7.5) --(10.5,8.5)--(8.5,8.5);
      \draw[pathline] (8.5,8.5) --(8.5,7.5) --(7.5,7.5) --(7.5,8.5) --(6.5,8.5) --(6.5,7.5) --(3.5,7.5) --(3.5,5.5) --(2.5,5.5) --(2.5,8.5)--(5.5,8.5);
      \draw[pathline] (5.5,8.5) --(5.5,9.5) --(11.5,9.5) --(11.5,7.5) --(12.5,7.5) --(12.5,12.5) --(11.5,12.5) --(11.5,10.5)--(10.5,10.5);
      \draw[pathline] (10.5,10.5) --(10.5,12.5) --(9.5,12.5) --(9.5,10.5) --(8.5,10.5) --(8.5,12.5) --(7.5,12.5) --(7.5,10.5) --(6.5,10.5) --(6.5,12.5)--(5.5,12.5);
      \draw[pathline] (5.5,12.5) --(5.5,10.5) --(4.5,10.5) --(4.5,9.5) --(3.5,9.5) --(3.5,11.5) --(4.5,11.5) --(4.5,12.5) --(2.5,12.5) --(2.5,9.5) --(2.5,9.5)--(1.5,9.5);
      \draw[pathline] (1.5,9.5) --(1.5,12.5) --(0.5,12.5) --(0.5,8.5) --(1.5,8.5) --(1.5,7.5) --(0.5,7.5) --(0.5,6.5) --(1.5,6.5) --(1.5,5.5) --(0.5,5.5)--(0.5,4.5);
      \draw[pathline] (0.5,4.5) --(1.5,4.5) --(1.5,3.5) --(0.5,3.5) --(0.5,2.5) --(1.5,2.5) --(1.5,1.5) --(0.5,1.5) --(0.5,0.5) --(4.5,0.5) --(5.5,0.5) --(5.5,1.5)--(2.5,1.5);
      \draw[pathline] (2.5,1.5) --(2.5,2.5) --(5.5,2.5) --(5.5,3.5) --(2.5,3.5) --(2.5,4.5) --(4.5,4.5) --(5.5,4.5) --(5.5,5.5) --(5.5,5.5)--(4.5,5.5);
      \draw[pathline] (4.5,5.5) --(4.5,6.5) --(6.5,6.5) --(6.5,5.5) --(7.5,5.5) --(7.5,6.5) --(8.5,6.5) --(8.5,4.5) --(11.5,4.5) --(11.5,5.5)--(12.5,5.5);
      \draw[pathline] (12.5,5.5) --(12.5,0.5) --(11.5,0.5) --(11.5,3.5) --(10.5,3.5) --(10.5,0.5) --(9.5,0.5) --(9.5,3.5) --(8.5,3.5) --(8.5,0.5) --(7.5,0.5)--(7.5,4.5);
      \draw[pathline] (7.5,4.5) --(6.5,4.5) --(6.5,0.5) --(6.5,0.0);
      \draw[pathline,->] (13,6.5)--(12.2,6.5);
      \draw[pathline,->] (6.5,.8)--(6.5,0);

      \tikzset{
        clearlabel/.style={font=\bfseries, text=black}
      }
      \node[clearlabel] at (0.25,12.75) {\fontsize{8}{12}\selectfont 2};
      \node[clearlabel] at (2.25,6.75) {\fontsize{8}{12}\selectfont 1};
      \node[clearlabel] at (5.5,8.75) {\fontsize{7}{12}\selectfont 136};
      \node[clearlabel] at (6.25,4.75) {\fontsize{8}{12}\selectfont 3};
    \end{tikzpicture}

    \subcaption{}
    \label{fig:VI}
  \end{subfigure}
  
  \caption{(continued) T-metacells for Remembered Length represent a vertex with 2 in-edges on opposite sides (d), and the path for each chosen in-edge (e-f)}

  \label{fig:remlen}

\end{figure}

\begin{figure}[htbp]
  \ContinuedFloat
  \centering

  \begin{subfigure}[b]{\textwidth}
    \centering
    \begin{tikzpicture}[scale=0.5]
      \draw[step=1, very thin] (0,0) grid (13,13);
      \draw[very thick] (0,0) rectangle (13,13);

      \draw[very thick] (2,6) rectangle (4,7);
      \draw[very thick] (6,3) rectangle (7,5);
      \draw[very thick] (6,5) rectangle (7,6);
      \draw[very thick] (6,6) rectangle (7,7);
      \draw[very thick] (7,6) rectangle (8,7);
      \draw[very thick] (8,6) rectangle (10,7);

      \tikzset{
        clearlabel/.style={text=black}
      }
      \node[clearlabel] at (0.5,12.5) {\fontsize{10}{12}\selectfont 2};
      \node[clearlabel] at (2.5,6.5) {\fontsize{8}{12}\selectfont 155};
      \node[clearlabel] at (6.5,6.5) {\fontsize{10}{12}\selectfont 1};
    \end{tikzpicture}

    \subcaption{}
  \end{subfigure}
  
  \vspace{1em}

  \begin{subfigure}[b]{0.45\textwidth}
    \centering
    \begin{tikzpicture}[scale=0.5]
      \draw[step=1, very thin] (0,0) grid (13,13);
      \draw[very thick] (0,0) rectangle (13,13);

      \draw[very thick] (2,6) rectangle (4,7);
      \draw[very thick] (6,3) rectangle (7,5);
      \draw[very thick] (6,5) rectangle (7,6);
      \draw[very thick] (6,6) rectangle (7,7);
      \draw[very thick] (7,6) rectangle (8,7);
      \draw[very thick] (8,6) rectangle (10,7);

      \tikzset{pathline/.style={blue, very thick,cap=round}}
      \draw[pathline] (0.0,6.5) --(3.5,6.5) --(3.5,5.5) --(0.5,5.5) --(0.5,4.5) --(3.5,4.5) --(3.5,3.5) --(0.5,3.5) --(0.5,2.5) --(3.5,2.5) --(3.5,1.5) --(0.5,1.5) --(0.5,0.5) --(5.5,0.5) --(5.5,1.5) --(4.5,1.5) --(4.5,2.5) --(5.5,2.5) --(5.5,3.5) --(4.5,3.5) --(4.5,3.5) --(4.5,4.5) --(5.5,4.5) --(5.5,5.5) --(5.5,5.5) --(5.5,5.5) --(4.5,5.5) --(4.5,6.5) --(5.5,6.5) --(5.5,7.5) --(0.5,7.5) --(0.5,8.5) --(5.5,8.5) --(5.5,9.5) --(0.5,9.5) --(0.5,10.5) --(5.5,10.5) --(5.5,11.5) --(0.5,11.5) --(0.5,12.5) --(6.5,12.5) --(12.5,12.5) --(12.5,11.5) --(6.5,11.5) --(6.5,10.5) --(6.5,7.5) --(7.5,7.5) --(7.5,10.5) --(8.5,10.5) --(8.5,7.5) --(9.5,7.5) --(9.5,10.5) --(10.5,10.5) --(10.5,7.5) --(11.5,7.5) --(11.5,10.5) --(12.5,10.5) --(12.5,0.5) --(7.5,0.5) --(7.5,1.5) --(11.5,1.5) --(11.5,2.5) --(7.5,2.5) --(7.5,3.5) --(11.5,3.5) --(11.5,4.5) --(7.5,4.5) --(7.5,5.5) --(11.5,5.5) --(11.5,6.5) --(6.5,6.5) --(6.5,0.5) --(6.5,0.0);

      \tikzset{
        clearlabel/.style={font=\bfseries, text=black}
      }
      \node[clearlabel] at (0.25,12.75) {\fontsize{8}{12}\selectfont 2};
      \node[clearlabel] at (2.5,6.75) {\fontsize{7}{12}\selectfont 155};
      \node[clearlabel] at (6.25,6.75) {\fontsize{8}{12}\selectfont 1};
      \draw[pathline, ->] (0,6.5)--(.8,6.5);
      \draw[pathline, ->] (6.5,.8)--(6.5,0);
    \end{tikzpicture}

    \subcaption{}
    \label{fig:VIII}
  \end{subfigure}\hfill
  \begin{subfigure}[b]{0.45\textwidth}
    \centering
    \begin{tikzpicture}[scale=0.5]
      \draw[step=1] (0,0) grid (13,13);
      \draw[very thick] (0,0) rectangle (13,13);

      \draw[very thick] (2,6) rectangle (4,7);
      \draw[very thick] (6,3) rectangle (7,5);
      \draw[very thick] (6,5) rectangle (7,6);
      \draw[very thick] (6,6) rectangle (7,7);
      \draw[very thick] (7,6) rectangle (8,7);
      \draw[very thick] (8,6) rectangle (10,7);

      \tikzset{pathline/.style={blue,very thick,cap=round}}
      \draw[pathline] (0.0,6.5) --(3.5,6.5) --(3.5,7.5) --(0.5,7.5) --(0.5,8.5) --(2.5,8.5) --(2.5,9.5) --(0.5,9.5) --(0.5,10.5)--(1.5,10.5);
      \draw[pathline] (1.5,10.5) --(2.5,10.5) --(2.5,11.5) --(0.5,11.5) --(0.5,12.5) --(3.5,12.5) --(12.5,12.5) --(12.5,11.5)--(3.5,11.5);
      \draw[pathline] (3.5,11.5) --(3.5,10.5) --(12.5,10.5) --(12.5,9.5) --(3.5,9.5) --(3.5,8.5) --(12.5,8.5) --(12.5,7.5)--(4.5,7.5);
      \draw[pathline] (4.5,7.5) --(4.5,6.5) --(5.5,6.5) --(5.5,5.5) --(0.5,5.5) --(0.5,4.5) --(5.5,4.5) --(5.5,3.5)--(0.5,3.5);
      \draw[pathline] (0.5,3.5) --(0.5,2.5) --(5.5,2.5) --(5.5,1.5) --(0.5,1.5) --(0.5,0.5) --(6.5,0.5) --(6.5,2.5) --(7.5,2.5) --(7.5,0.5)--(8.5,0.5);
      \draw[pathline] (8.5,0.5) --(8.5,1.5) --(8.5,2.5) --(9.5,2.5) --(9.5,0.5) --(10.5,0.5) --(10.5,2.5) --(11.5,2.5) --(11.5,0.5)--(12.5,0.5);
      \draw[pathline] (12.5,0.5) --(12.5,5.5) --(11.5,5.5) --(11.5,3.5) --(10.5,3.5) --(10.5,5.5) --(9.5,5.5) --(9.5,3.5)--(8.5,3.5);
      \draw[pathline] (8.5,3.5) --(8.5,5.5) --(7.5,5.5) --(7.5,3.5) --(6.5,3.5) --(6.5,6.5) --(12.5,6.5) --(13.0,6.5);
      \draw[pathline, ->] (0,6.5)--(.8,6.5);
      \draw[pathline, ->] (12.2,6.5)--(13,6.5);

      \tikzset{
        clearlabel/.style={font=\bfseries, text=black}
      }
      \node[clearlabel] at (0.25,12.75) {\fontsize{8}{12}\selectfont 2};
      \node[clearlabel] at (2.5,6.75) {\fontsize{7}{12}\selectfont 155};
      \node[clearlabel] at (6.25,6.75) {\fontsize{8}{12}\selectfont 1};
    \end{tikzpicture}

    \subcaption{}
    \label{fig:IX}
  \end{subfigure}

  \caption{(continued) T-metacells for Remembered Length represent a vertex with 2 out-edges on adjacent sides (g), and the path for each chosen out-edge (h-i)}

  \label{fig:remlen}
\end{figure}

\begin{figure}[htbp]
  \ContinuedFloat
  \centering

  \begin{subfigure}[b]{\textwidth}
    \centering
    \begin{tikzpicture}[scale=0.5]
      \draw[step=1] (0,0) grid (13,13);
      \draw[very thick] (0,0) rectangle (13,13);

      \draw[very thick] (3,6) rectangle (5,7);
      \draw[very thick] (5,6) rectangle (6,7);
      \draw[very thick] (6,2) rectangle (7,4);
      \draw[very thick] (6,6) rectangle (7,7);
      \draw[very thick] (7,6) rectangle (8,7);
      \draw[very thick] (8,6) rectangle (10,7);

      \tikzset{
        clearlabel/.style={text=black}
      }
      \node[clearlabel] at (0.5,12.5) {\fontsize{10}{12}\selectfont 2};
      \node[clearlabel] at (6.5,2.5) {\fontsize{8}{12}\selectfont 155};
      \node[clearlabel] at (6.5,6.5) {\fontsize{10}{12}\selectfont 1};
    \end{tikzpicture}

    \subcaption{}
  \end{subfigure}

  \vspace{1em}

  \begin{subfigure}[b]{0.45\textwidth}
    \centering
    \begin{tikzpicture}[scale=0.5]
      \draw[step=1] (0,0) grid (13,13);
      \draw[very thick] (0,0) rectangle (13,13);

      \draw[very thick] (3,6) rectangle (5,7);
      \draw[very thick] (5,6) rectangle (6,7);
      \draw[very thick] (6,2) rectangle (7,4);
      \draw[very thick] (6,6) rectangle (7,7);
      \draw[very thick] (7,6) rectangle (8,7);
      \draw[very thick] (8,6) rectangle (10,7);

      \tikzset{pathline/.style={blue, very thick,cap=round}}
      \draw[pathline] (0.0,6.5) --(9.5,6.5) --(11.5,6.5) --(11.5,7.5) --(0.5,7.5) --(0.5,8.5) --(11.5,8.5) --(11.5,9.5) --(1.5,9.5) --(0.5,9.5) --(0.5,10.5) --(11.5,10.5) --(11.5,11.5) --(0.5,11.5) --(0.5,11.5) --(0.5,12.5) --(12.5,12.5) --(12.5,0.5) --(11.5,0.5) --(11.5,5.5) --(10.5,5.5) --(10.5,0.5) --(9.5,0.5) --(9.5,5.5) --(8.5,5.5) --(8.5,0.5) --(7.5,0.5) --(7.5,5.5) --(6.5,5.5) --(6.5,4.5) --(5.5,4.5) --(5.5,5.5) --(5.5,5.5) --(4.5,5.5) --(4.5,5.5) --(4.5,4.5) --(3.5,4.5) --(3.5,5.5) --(2.5,5.5) --(2.5,4.5) --(1.5,4.5) --(1.5,5.5) --(0.5,5.5) --(0.5,0.5) --(1.5,0.5) --(1.5,3.5) --(2.5,3.5) --(2.5,0.5) --(3.5,0.5) --(3.5,3.5) --(4.5,3.5) --(4.5,0.5) --(5.5,0.5) --(5.5,3.5) --(6.5,3.5) --(6.5,0.5) --(6.5,0.0);

      \tikzset{
        clearlabel/.style={font=\bfseries, text=black}
      }
      \node[clearlabel] at (0.25,12.75) {\fontsize{8}{12}\selectfont 2};
      \node[clearlabel] at (6.5,2.75) {\fontsize{7}{12}\selectfont 155};
      \node[clearlabel] at (6.25,6.75) {\fontsize{8}{12}\selectfont 1};
      \draw[pathline, ->] (0.8,6.5)--(0,6.5);
      \draw[pathline, ->] (6.5,0)--(6.5,.8);
    \end{tikzpicture}

    \subcaption{}
    \label{fig:XI}
  \end{subfigure}\hfill
  \begin{subfigure}[b]{0.45\textwidth}
    \centering
    \begin{tikzpicture}[scale=0.5]
      \draw[step=1] (0,0) grid (13,13);
      \draw[very thick] (0,0) rectangle (13,13);

      \draw[very thick] (3,6) rectangle (5,7);
      \draw[very thick] (5,6) rectangle (6,7);
      \draw[very thick] (6,2) rectangle (7,4);
      \draw[very thick] (6,6) rectangle (7,7);
      \draw[very thick] (7,6) rectangle (8,7);
      \draw[very thick] (8,6) rectangle (10,7);

      \tikzset{pathline/.style={blue, very thick,cap=round}}
      \draw[pathline] (13.0,6.5) --(1.5,6.5) --(1.5,7.5) --(12.5,7.5) --(12.5,8.5) --(1.5,8.5) --(1.5,9.5) --(12.5,9.5) --(12.5,10.5) --(1.5,10.5) --(1.5,11.5) --(12.5,11.5) --(12.5,12.5) --(0.5,12.5) --(0.5,0.5) --(1.5,0.5) --(1.5,5.5) --(2.5,5.5) --(2.5,0.5) --(3.5,0.5) --(3.5,5.5) --(4.5,5.5) --(4.5,0.5) --(5.5,0.5) --(5.5,5.5) --(6.5,5.5) --(6.5,4.5) --(7.5,4.5) --(7.5,5.5) --(8.5,5.5) --(8.5,4.5) --(9.5,4.5) --(9.5,5.5) --(10.5,5.5) --(10.5,4.5) --(11.5,4.5) --(11.5,5.5) --(12.5,5.5) --(12.5,0.5) --(12.5,0.5) --(11.5,0.5) --(11.5,3.5) --(10.5,3.5) --(10.5,0.5) --(9.5,0.5) --(9.5,3.5) --(8.5,3.5) --(8.5,1.5) --(8.5,0.5) --(7.5,0.5) --(7.5,3.5) --(7.5,3.5) --(6.5,3.5) --(6.5,0.5) --(6.5,0.0);

      \tikzset{
        clearlabel/.style={font=\bfseries, text=black}
      }
      \node[clearlabel] at (0.25,12.75) {\fontsize{8}{12}\selectfont 2};
      \node[clearlabel] at (6.5,2.75) {\fontsize{7}{12}\selectfont 155};
      \node[clearlabel] at (6.25,6.75) {\fontsize{8}{12}\selectfont 1};
      \draw[pathline, ->] (12.2,6.5)--(13,6.5);
      \draw[pathline, ->] (6.5,0)--(6.5,.8);
    \end{tikzpicture}

    \subcaption{}
    \label{fig:XII}
  \end{subfigure}

  \caption{(continued) T-metacells for Remembered Length represent a vertex with 2 out-edges on opposite sides (j), and the path for each chosen out-edge (k-l)}

  \label{fig:remlen}
\end{figure}

\clearpage

\section{Conclusion and Future Work}
We proved the NP-completeness of three puzzles, All or Nothing, Water Walk, and Remembered Length (with its generalized version) via T-metacell reductions from the grid graph Hamiltonicity problem. While the underlying Hamiltonicity problem is ASP-complete, our constructed gadgets do not enforce a unique path between each pair of exits. Establishing such uniqueness would imply the ASP-completeness of these puzzles, and thus remains an interesting direction for future work.

\bibliography{references}
\end{document}